\documentclass[sigconf]{acmart}
\renewcommand\footnotetextcopyrightpermission[1]{} 
\usepackage{booktabs} 
\usepackage{xspace}
\usepackage{mathtools}
\usepackage{amsmath,amssymb}
\usepackage{enumitem}
\usepackage{verbatim}
\usepackage{ragged2e}
\usepackage{algorithm}
\usepackage{algorithmicx}
\usepackage{algcompatible}
\usepackage{algpseudocode}
\algrenewcommand\algorithmicindent{0.7em}
\usepackage[toc,page]{appendix}
\usepackage{graphicx} 
\usepackage{color}
\usepackage{caption}
\usepackage{ragged2e}
\usepackage{tabularx}
\usepackage{url}
\usepackage{hyperref}
\usepackage{cleveref}
\usepackage{pgfplotstable}
\pgfplotsset{compat=1.11,
        /pgfplots/ybar legend/.style={
        /pgfplots/legend image code/.code={%
        \draw[##1,/tikz/.cd,bar width=3pt,yshift=-0.2em,bar shift=0pt]
                plot coordinates {(0cm,0.8em)};},
},
}
\hypersetup{
     colorlinks   = true,
     citecolor    = black
}
\usepackage{amsmath}
\usepackage{mathtools}
\usepackage{amsmath,amssymb,latexsym} 
\usepackage{float}
\usepackage{breqn}
\usepackage{xcolor}
\usepackage{amsthm}
 \usepackage{relsize}
\usepackage[justification=centering]{caption}
\usepackage{caption}
\usepackage{subcaption}
\usepackage{pgf,tikz}
\usetikzlibrary{matrix}
\usepackage{pgfplots}
\pgfplotsset{width=7cm,compat=1.8}
\usetikzlibrary{arrows}
\usetikzlibrary{shadows}
\usepackage[utf8]{inputenc}
\usepackage[english]{babel}
\usepackage{xcolor, soul}
\usepackage{multicol}

\newcommand{\punt}[1]{}
\newcommand{\cmnt}[1]{}

\newcommand{\cgds} {concurrent graph data structure\xspace}

\newcommand{\lble} {linearizable\xspace}
\newcommand{\lbty} {linearizability\xspace}
\newcommand{\rbty} {reachability\xspace}

\newcounter{history}

\newcommand{\figref}[1]{Figure~\ref{fig:#1}}

\newcommand{\ignore}[1]{}

\algdef{SE}[DOWHILE]{Do}{doWhile}{\algorithmicdo}[1]{\algorithmicwhile\ #1}%
\newcommand{\mth} {operation\xspace}
\newcommand{\cc} {correctness-criterion\xspace}

\newcommand{\lp} {LP\xspace}

\newcommand{\tru} {\texttt{true}\xspace}
\newcommand{\fal} {\texttt{false}\xspace}
\newcommand{\nul} {\texttt{NULL}\xspace}

\newcommand{\vnodes} {{\tt VNodes}\xspace}

\newcommand{\enode}{{\tt ENode}\xspace}
\newcommand{\vnode}{{\tt VNode}\xspace}
\newcommand{\bfsnode}{{\tt BFSNode}\xspace}
\newcommand{\bfstree}{{\tt BFS-tree}\xspace}
\newcommand{\vlist} {vertex-list\xspace}
\newcommand{\elist} {edge-list\xspace}
\newcommand{\elists} {edge-lists\xspace}

\newcommand{\vh}{\texttt{vh}\xspace}
\newcommand{\vt}{\texttt{vt}\xspace}
\newcommand{\eh}{\texttt{eh}\xspace}
\newcommand{\et}{\texttt{et}\xspace}
\newcommand{\addv}{\textsc{AddVertex}\xspace}
\newcommand{\remv}{\textsc{RemoveVertex}\xspace}
\newcommand{\adde}{\textsc{AddEdge}\xspace}
\newcommand{\reme}{\textsc{RemoveEdge}\xspace}
\newcommand{\conv}{\textsc{ContainsVertex}\xspace}
\newcommand{\cone}{\textsc{ContainsEdge}\xspace}

\newcommand{\locvplus}{\textsc{locV}\xspace}
\newcommand{\loceplus}{\textsc{locE}\xspace}
\newcommand{\loccplus}{\textsc{locC}\xspace}
\newcommand{\concplus}{\textsc{ConCPlus}\xspace}

\newcommand{\createe} {\textsc{CEnode}\xspace}
\newcommand{\createv}{\textsc{CVnode}\xspace}
\newcommand{\convplus}{\textsc{ConVPlus}\xspace}

\newcommand{\fadd}{FetchAndAdd\xspace}

\newcommand{\getpath}{\textsc{GetPath}\xspace}

\newcommand{\treec}{\textsc{TreeCollect}\xspace}

\newcommand{\checkvisited}{\textsc{CheckVisited}\xspace}
\newcommand{\visitedarray}{\texttt{VisitedArray}\xspace}

\newcommand{\of}{obstruction-free\xspace}
\newcommand{\Nbk}{Non-blocking\xspace}
\newcommand{\nbk}{non-blocking\xspace}

\newcommand{\vcs}{vertices\xspace}

\newcommand{\cas}{compare-and-swap\xspace}
\newcommand{\CAS}{\texttt{CAS}\xspace}
\newcommand{\faa}{fetch-and-add\xspace}
\newcommand{\FAA}{\texttt{FAA}\xspace}

\newcommand{\acadde}{\textsc{AcyclicAddEdge\xspace}}

\newcommand{\accone}{\textsc{AcyclicContainsEdge\xspace}}

\newcommand{\lf}{lock-free\xspace}

\newcommand{\enext}{{\tt enxt}\xspace}
\newcommand{\vnext}{{\tt vnxt}\xspace}
\newcommand{\bnext}{{\tt nxt}\xspace}

\newcommand{\pointv}{{\tt ptv}\xspace}
\newcommand{\ecount}{{\tt ecnt}\xspace}
\newcommand{\lecount}{{\tt lecnt}\xspace}

\newcommand{\vntp}{{\tt VERTEX NOT PRESENT}\xspace}
\newcommand{\entp}{{\tt EDGE NOT PRESENT}\xspace}
\newcommand{\ventp}{{\tt VERTEX OR EDGE NOT PRESENT}\xspace}

\newcommand{\ep}{{\tt EDGE PRESENT}\xspace}
\newcommand{\eadd}{{\tt EDGE ADDED}\xspace}
\newcommand{\er}{{\tt EDGE REMOVED}\xspace}
\newcommand{\ef}{{\tt EDGE FOUND}\xspace}

\newcommand{\scan}{\textsc{Scan}\xspace}

\newcommand{\comparepath}{\textsc{ComparePath}\xspace}
\newcommand{\comparetree}{\textsc{CompareTree}\xspace}

\newcommand{\isMarked}{\textsc{isMrkd}\xspace}
\newcommand{\MarkedRef}{\textsc{MrkdRf}\xspace}
\newcommand{\unMarkedRef}{\textsc{UnMrkdRf}\xspace}
\definecolor{butter1}{rgb}{0.988,0.914,0.310}
\definecolor{chocolate1}{rgb}{0.914,0.725,0.431}
\definecolor{chameleon1}{rgb}{0.541,0.886,0.204}
\definecolor{skyblue1}{rgb}{0.447,0.624,0.812}
\definecolor{plum1}{rgb}{0.678,0.498,0.659}
\definecolor{scarletred1}{rgb}{0.937,0.161,0.161}

\acmConference[ICDCN 2019]{ICDCN conference}{Bangalore}{India}
\acmYear{2019}
\copyrightyear{2019}

\begin{document}
\title{A Simple and Practical Concurrent \Nbk Unbounded Graph with Linearizable Reachability Queries}

\author{Bapi Chatterjee\texorpdfstring{{$^\ast$}}, Sathya Peri\texorpdfstring{{$^\dagger$}}, Muktikanta Sa\texorpdfstring{{$^\dagger$}}, Nandini Singhal\texorpdfstring{{$^\ddagger$}}\footnotemark{$^\nmid$}}
\affiliation{%
  \institution{{$^\ast$}IBM, India Research Lab, New Delhi, India, bapchatt@in.ibm.com}
  \institution{{$^\dagger$}Department of Computer Science \& Engineering, \\ Indian Institute of Technology Hyderabad, India, \{sathya\_p, cs15resch11012\}@iith.ac.in}
  \institution{{$^\ddagger$}Microsoft (R\&D) Pvt. Ltd., Bangalore, India, nasingha@microsoft.com}
}
\begin{abstract}
Graph algorithms applied in many applications, including social networks, communication networks, VLSI design, graphics, and several others, require dynamic modifications -- addition and removal of vertices and/or edges -- in the graph. This paper presents a novel concurrent non-blocking algorithm to implement a dynamic unbounded directed graph in a shared-memory machine. The addition and removal operations of vertices and edges are lock-free. For a finite sized graph, the lookup operations are wait-free. Most significant component of the presented algorithm is the \rbty query in a concurrent graph. The \rbty queries in our algorithm are obstruction-free and thus impose minimal additional synchronization cost over other operations. We prove that each of the data structure operations are linearizable. We extensively evaluate a sample C/C++ implementation of the algorithm through a number of micro-benchmarks. The experimental results show that the proposed algorithm scales well with the number of threads and on an average provides $5$ to $7$x performance improvement over a concurrent graph implementation using coarse-grained locking. 
\end{abstract}


\maketitle
\footnotetext{$^\nmid$Work done while a student at IIT Hyderabad.}

\section{Introduction}\subsection{Background}
A \textit{graph} is a highly useful data structure that models the pairwise relationships among real-world objects. Formally, it is represented as an ordered pair $G = (V,E)$, where $V$ is the set of \vcs and $E$ is the set of edges. They underlay a number of important applications such as various kinds of networks (social, semantic, genomics, etc.), VLSI design, graphics, and several others. Generally, these applications require modifications such as insertion and deletion of \vcs and/or edges in the graph to make them \textit{dynamic}. Furthermore, they also require the data structure to grow at the run-time. The rise of multi-core systems has necessitated efficient and correct design of concurrent data structures, which can take advantage of the multi-threaded implementations. Naturally, a large number of applications would significantly benefit from concurrent dynamic unbounded graphs implemented on ubiquitous multi-core computers.


Arguably, the most important application of such a graph is performing a \textit{\rbty query}: for a given pair of vertices $u, v \in V$, determine if a sequence of adjacent vertices, i.e. a path, exists in $V$ that starts at $u$ and ends at $v$. In many use-cases, a reachability query requires returning the path if it exists. Performing a \rbty query effectively entails exploring all the possible paths starting at $u$. In a dynamic and concurrent setting, where both $V$ and $E$ can be modified by a concurrent thread, it is extremely challenging to return a \textit{valid} path, or for that matter being assured that there does not exist a path at the return of a \rbty query. Besides, deletion of a vertex $v \in V$ involves deleting it along with all its incoming and outgoing edges in $E$. Obviously, a look-up or an insertion operation with a concurrent deletion poses complex design issues with regards to their correctness.

A well-accepted \cc for concurrent data structures is \emph{\lbty} \cite{HerlWing:1990:TPLS}. Broadly, a provably linearizable operation is perceived by a user as if it takes effect instantaneously at a point between its invocation and response during any concurrent execution. A simple and popular approach to handle updates in a concurrent data structure, while ensuring \lbty, is by way of mutual exclusion using \textit{locks}. In case of a \rbty query, that would essentially amount to locking the entire graph at its invocation and releasing the lock only at the return. However, in an asynchronous shared-memory system, such an implementation is vulnerable to arbitrary delays due to locks acquired by a slow thread in addition to several other pitfalls such as deadlock, priority inversion and convoying. 

In contrast, the \textit{non-blocking} implementations -- \textit{wait-free}, \textit{lock-free}, and \textit{obstruction-free} \cite{MauriceNir, Herlihy} -- do not use locks and therefore are free from the drawbacks mentioned above. A wait-free operation in a concurrent data structure always finishes in a finite number of steps taken by a non-faulty thread. However, such a strong progress guarantee is too costly to implement.  
Alternatively, a lock-free operation is guaranteed to finish in a finite number of steps taken by some non-faulty thread. A concurrent data structure with all its operations satisfying lock-freedom is generally scalable. However, ensuring that each of the operations in a complex data structure, such as a graph, finishes in a finite number of steps by some thread, is still extremely challenging. Though weaker than the lock-freedom, the obstruction-freedom is still good enough to avoid the pitfalls of locks: an obstruction-free operation always finishes in a finite number of steps taken by any non-faulty thread running in isolation. Importantly, for obstruction-free operations, the design complexity as well as the synchronization overhead are much lower in comparison to their lock-free counterparts.

\subsection{Our contributions}
In this paper, we present a novel \nbk algorithm for an unbounded directed\footnote{An undirected graph can be directly extended from a directed graph.} graph data structure. The contributions of this work are summarized below:

\begin{enumerate}
	\item For a directed graph $G = (V,E)$, we describe an Abstract Data Type (ADT) comprising of \textit{modifications} --\textit{addition}, \textit{removal} -- and \textit{lookup} operations on the sets $V$ and $E$. The ADT also comprises an operation to perform the \textit{\rbty queries} in $G$. The graph is represented as an adjacency list which enables it to grow without bound (up to the availability of memory) and sink at the runtime. (\Cref{sec:model})

	\item We describe an efficient concurrent \nbk implementation of the ADT (\Cref*{sec:ds-design}). To our knowledge, this is the first work on a \nbk unbounded graph. The spotlight of our work is an \of \rbty query. (\Cref{sec:getPath})
	
	\item We prove the correctness in terms of the linearizability of the ADT operations. We also prove the non-blocking progress guarantee: (a) the modifications and lookup on vertices and edges of the graph are \lf, (b) the \rbty queries are \of, and (c) if the graph size is finite, the vertex and edge lookup operations are wait-free.(\Cref{sec:proof})
	
	\item For an experimental validation, we implemented the algorithm in C/C++. For comparison, we also implemented a sequential and a coarse-grained lock-based concurrent graph. We tested the implementations using a number of micro-benchmarks simulating varying workloads. The experiments demonstrate that the lock-free algorithm is highly scalable with the number of threads. We observed up to $5 - 7$x higher throughput utilizing the available threads in our multi-core workstation while comparing the \nbk algorithm with its sequential and coarse-grained lock-based counterparts. (\Cref{sec:results}) 
\end{enumerate}

\subsection{Related work} 
The concurrent graph data structure is largely an unexplored topic. There has been a recent interesting and relevant work by Kallimanis and Kanellou \cite{Kallimanis}. They proposed a concurrent graph that supports wait-free edge modifications and traversals. Their algorithm works on an adjacency matrix representation. They do not allow addition or removal of vertices after initialization of the graph, which renders it unsuitable for many applications that require dynamic and unbounded graphs. Moreover, it is not clear how the wait-free traversal will return a set of adjacent vertices in their algorithm.

\subsection{A brief overview of the design}
We implement the adjacency list of a directed graph as a sorted linked-list of \textit{vertex-nodes}, where each of the vertex-node roots a sorted linked-list of \textit{edge-nodes}, see \Cref{fig:conGraph}. The edge-nodes maintain pointers to the corresponding vertex-nodes to enable efficient graph traversals. The individual edge-node-lists and the vertex-node-list are lock-free with regards to the modifications and lookup operations.

The lock-free operations in the graph intuitively appear as a composition of the lock-free operations in the sorted vertex-list and the edge-lists. However, it is well-known that the lock-freedom is not composable \cite{dang2011progress}, the progress guarantee of our algorithm is proved independent of the lock-free property of the component linked-lists. Furthermore, we propose some elegant optimizations in the operations' synchronization that not only ensure provable \lbty but also bring simplicity in the design.

For \rbty queries we perform breadth first search (BFS) traversals in the graph. We implement the BFS traversals fully non-recursively for efficiency in a concurrent setting. It is natural that a \rbty query is much more costlier compared to a modification or a lookup operation. However, in a concurrent setting it needs to synchronize with other concurrent operations. To ensure that the overall performance does not suffer from large \rbty queries, we do not employ other operations to help them. Instead, to achieve the \lbty, we repeatedly collect concise versions of the graph and validate them by matching the return of two consecutive collections. Our approach is essentially based on the \emph{double collect} \cite[Chap 4]{MauriceNir} aided with several interesting optimizations. This design choice results in \of progress guarantee for the \rbty queries.

\label{sec:intro}

\section{Preliminaries}\label{sec:model}
\subsection{The ADT}

An abstract directed graph is given as $G = (V,E)$, where $V$ is the set of \textit{\vcs} and $E$ is the set of directed \textit{edges} (ordered pair of \vcs). Each edge connects an ordered pair of \vcs belonging to $V$. A $v \in V$ maintains an immutable unique key $k \in K$, where $K$ is a totally ordered set. A vertex $v \in V$ with key $k$ is denoted by $v(k)$. We use $e(k,l)$ to denote an edge $(v(k),v(l)) \in E$.


We define an ADT for operations on $G$ as given below.  
\begin{enumerate}
	\item The $\addv(k)$ operation adds a vertex $v(k)$ to $V$,  if $v(k) \notin V$ and returns \tru. If $v(k) \in V$, it returns \fal. 
	\item The $\remv(k)$ operation removes $v(k)$ from $V$, if $v(k) \in V$ and returns \tru. If $v(k) \notin V$, it returns \fal. A successful $\remv(k)$ ensures that all $(j,k), (k,l) \in E$ are removed as well.
	\item The $\conv(k)$ returns \tru, if $v(k) \in V$; otherwise, it returns \fal.
	\item The $\adde(k,l)$ operation adds an edge $e(k,l)$ to $E$, if (a) $e(k,l) \notin E$, (b) $v(k) \in V$, and (c) $v(l) \in V$. If either of the conditions (a), (b) or (c) not satisfied, no change is made in $E$. For clarity about the reason of failure in adding an edge, we use different \textit{indicative strings} for the different return cases.
	\item The $\reme(k,l)$ operation removes the edge $e(k,l)$ if $e(k,l) \in E$. If $e(k,l) \notin E$, it makes no change in $E$. Similar to \adde, a \reme returns strings indicating if $v(k) \notin V$ or $v(l) \notin V$ or $e(k,l) \notin E$.
	\item The $\cone(k,l)$ operation returns an indicative string "\ep" if $e(k,l) \in E$; otherwise, it returns similar strings as a \reme.
	\item The $\getpath(k,l)$ operation returns a sequence of vertices -- called a \textit{path} -- $\{v_i\}_{i=1}^p \subseteq V$, if (a) $v(k) \in V$, (b) $v(l) \in V$ (c) $(v(k),v_1) \in E$, (c) $(v_p,v(l)) \in E$, and (d) $(v_i,v_{i+1}) \in E~ \forall 1 \leq i \leq p$; otherwise, it returns \nul.
\end{enumerate}

\subsection{The data structure}\label{subsec:ds}
\begin{figure}[H]
	\captionsetup{font=footnotesize}
	\centerline{\scalebox{0.5}{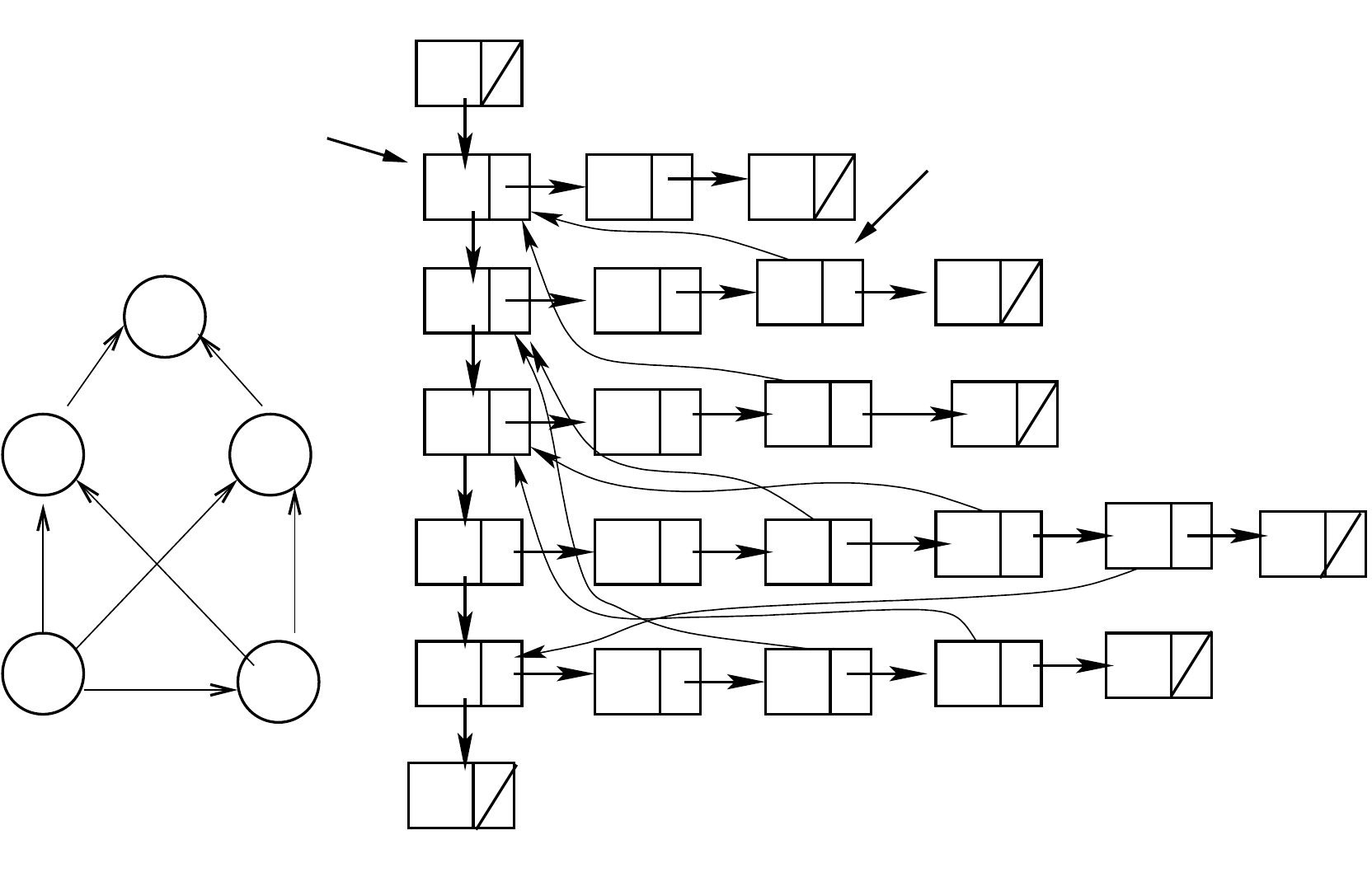}}
	\caption{(a) A directed Graph (b) The \cgds for (a).}
	\label{fig:conGraph}
\end{figure}
The above ADT is implemented by a data structure based on dynamically allocated \textit{nodes}. Nodes are connected by  word-sized \textit{pointers}. The data structure essentially implements the \textit{adjacency list} of $G$. See \Cref{fig:conGraph}. In our design an adjacency list is a composite linked-list of linked-lists. We call the composite linked-list the \textit{\vlist}. Each node in the \vlist, called a \textit{\vnode}, corresponds to a $v(k) \in V$ and roots a linked-list called an \textit{\elist}. An \elist rooted at a \vnode corresponding to $v(k) \in V$ consists of nodes, called \textit{\enode}, that correspond to the directed edges $e(k,l) \in E$. 

The structures of the \vnode and \enode are given in \Cref{fig:struct-evnode}. A \vnode $v(k)$ consists of two pointers \vnext and \enext in addition to an immutable key $k$. \vnext points to the next \vnode in the \vlist, whereas, \enext points to the head of the \elist rooted at $v(k)$. It also contains an array \visitedarray and an atomic counter \ecount. The array and the counter are used to facilitate BFS traversals, which we describe later. In the \elist of $v(k)$, an \enode $e(k,l)$ contains two pointers \pointv and \enext in addition to an immutable key field $l$. We do not need to record the key of $v(k)$ in an \enode $e(k,l)$ because the \enode{s} are confined to the context of a single \elist. \pointv of $e(k,l)$ points to the \vnode $v(l)$, whereas, its \enext points to the next \enode in the \elist.
\begin{figure}[h]
	\captionsetup{font=footnotesize}
	\begin{subfigure}{.33\textwidth}	
	\begin{footnotesize}
		\begin{tabbing}
			\hspace{0.1in} \= \hspace{0.1in} \= \hspace{0.1in} \=  \hspace{0.1in} \= \\
			\> {\bf class \vnode \{} \\
			\> \> \texttt{int}  $k$; 
			\\
			\> \> {\vnode~ \vnext;} 
			\\
			\> \> {\enode~ \enext;} 
			 \\
			\> \> {int \visitedarray [];} 
			\\
			\> \> {int \ecount;} 
			\\
			\> \} 
			\end{tabbing}
					\vspace{-0.3in}
				\end{footnotesize}
			\end{subfigure}
			\begin{subfigure}{.32\textwidth}	
			\begin{footnotesize}
					\begin{tabbing}
						\hspace{0.1in} \= \hspace{0.1in} \= \hspace{0.1in} \=  \hspace{0.1in} \= \\
			\\
			\> {\bf class \enode \{} \\
			\> \> \texttt{int}  $l$;   
			\\
			\> \> {\vnode~ \pointv;} 
			\\
			\> \> {\enode~ \enext;} 
			\\
			\> \}\\
		\end{tabbing}
		\vspace{-0.3in}
	\end{footnotesize}
	\end{subfigure}
		\begin{subfigure}{.32\textwidth}	
				\begin{footnotesize}
						\begin{tabbing}
					\hspace{0.1in} \= \hspace{0.1in} \= \hspace{0.1in} \=  \hspace{0.1in} \= \\
							\> {\bf class \bfsnode \{} \\
							\> \> {\tt \vnode $n$;} 
							\\
							\> \> {int \lecount;} 
			
							\\
							\> \> {\tt \bfsnode~ \bnext;} 
							\\
							\> \> {\tt \bfsnode $p$;} 
							
							\\
							\> \} \\
                            \\
			\end{tabbing}
			\vspace{-0.3in}
		\end{footnotesize}
		\end{subfigure}
	\caption{Structure of \enode, \vnode and \bfsnode.}
	\label{fig:struct-evnode}
\end{figure}

To avoid dereferencing the null pointers, we initialize the \vlist with sentinel nodes $v(-\infty)$ and $v(\infty)$ representing dummy head and tail, respectively. Similarly, each of the \elists are initialized with $e(-\infty)$ and $e(\infty)$ representing its dummy head and tail, respectively. 

The data structure maintains simple \textit{invariants}: (a) the \vlist is sorted, and (b) each of the \elists are sorted. 

The implementation of the ADT operations in the data structure as described here is called their \textit{sequential specification}. Traversals in the \vlist and \elists are performed following their sorted order starting from the dummy head and stopping at an appropriate node. $\addv(k)$, $\remv(k)$ and $\conv(k)$ locate $v(k)$ traversing down the \vlist starting from $v(-\infty)$ and stopping at a node $v(j)$ such that $j \leq k$ and $\nexists ~v(j') \in V$ where $j < j' \leq k$. The pointer modifications by \addv and \remv are exactly same as in a sequential linked-list data structure \cite{Harrisdisc01}. \conv returns \tru if $j = k$ else it returns \fal.

The $\adde(k,l)$, $\reme(k,l)$ and $\cone(k,l)$ first locate both $v(k)$ and $v(l)$ in the \vlist. Only if both the nodes are located in the \vlist, the next step is taken. In the next step they traverse down the \elist of $v(k)$. Addition, removal and lookup in an \elist are similar to those in the \vlist. Interestingly, when the \vnode $v(l)$ is removed from the \vlist, we may not necessarily have to remove $e(k,l)$ from the  \elist of $v(k)$, and still a $\cone(k,l)$ will work correctly because it would not traverse down the \elist of $v(k)$ after finding that $v(l)$ does not exist in the \vlist. The return of these operations are indicative strings as mentioned in the ADT definition.

A $\getpath(k,l)$ operation first locates $v(k)$ and $v(l)$ in the \vlist. If either of them is not located, it right away returns \nul. On locating $v(k)$, it performs a BFS traversal in the data structure starting at $v(k)$ \cite{cormen2009introduction}. The pointers \pointv of \enode{s} help a traversal to directly go to a \vnode from an \enode. An integer (for a dirty bit as described in \cite{cormen2009introduction}) is used to keep track of the visited nodes during BFS traversals. The array \visitedarray in a \vnode serves this purpose in concurrent multi-threaded implementations. A BFS traversal for $\getpath(k,l)$ returns as soon as the node $v(l)$ is located. 

The return of a BFS is a linked-list of \bfsnode{s}, see \Cref{fig:struct-evnode}. A \bfsnode $b(v)$, corresponding to a visited \vnode $v$, contains $v$. A counter \lecount, to make a local copy of the vertex \ecount.  It also contains two \bfsnode pointers \bnext and $p$. The pointer \bnext points to the next node in the returned linked-list. $p$ is a special pointer pointing to the \bfsnode containing the \vnode in the graph from which the BFS traverses to $v$. $p$ facilitates the tracing of the path from $v(k)$ to $v(l)$ in the return of $\getpath(k,l)$. Starting from the last node in the returned list, which necessarily corresponds to $v(l)$ if there exists a path between $v(k)$ and $v(l)$, $\getpath(k,l)$ returns the sequence of the adjacent nodes. If the last \bfsnode in the returned list does not contain $v(l)$, the return of $\getpath(k,l)$ is \nul.

\section{The lock-free algorithm}\label{sec:ds-design}
\subsection{The shared-memory system}
For a concurrent non-blocking implementation of the data structure, we consider a \textit{shared-memory system} consisting of a finite set of \textit{processors} accessed by a finite set of \textit{threads} that run in a completely asynchronous manner. The threads communicate with each other by invoking \mth{s} on shared objects and getting corresponding responses. The pointers and other fields of the various nodes are implemented by the shared objects. The system supports atomic \texttt{read}, \texttt{write}, \texttt{\faa} (\FAA)  and \texttt{\cas} (\CAS) instructions. 
\begin{figure*}[!htp]
\captionsetup{font=footnotesize}
	\begin{subfigure}{.3\textwidth}	
\begin{algorithmic}[1]
\renewcommand{\algorithmicprocedure}{\textbf{Operation}}	
\small
	\Procedure{ \addv($k$)}{}\label{advstart}
	\While{(\tru)}
	\State{$\langle pv, cv \rangle$ $\gets$ \locvplus($\vh$, $k$);}
	\If{($cv.k$ $=$ $k$)} {return {\fal;}}
	\Else
	\State{$nv$ $\gets$ \createv($k$); $nv.\vnext$ $\gets$ $cv$;}
	\If{(CAS($pv.\vnext$, $cv, nv$))} \label{lin:cas-addv}
	\State return {\tru;} 
	\EndIf
	\EndIf
	\EndWhile
	\EndProcedure\label{advend}
	\algstore{addv}
\end{algorithmic}	
    \hrule
    \begin{algorithmic}[1]
\renewcommand{\algorithmicprocedure}{\textbf{Operation}}	
	\algrestore{addv}
	\small
	\Procedure{ $\remv$($k$)}{}\label{remvstart}
	\While{(\tru)}
	\State{$\langle pv, cv \rangle$ $\gets$ \locvplus($\vh$, $k$);}
	\If{($cv.k$ $\neq$ $k$)}
	\State return {\fal;}
	\EndIf
	\State{$cn$ $\gets$ $cv.\vnext$;}
	\If{($\neg$ \isMarked($cn$))}
	\If{(CAS($cv.\vnext$, $cn$, \MarkedRef($cn$)))} \label{lin:cas-remv} 
	\If{(CAS($pv.\vnext, cv, cn$))} \label{lin:cas-remv-phy}
	\State{ $break$;} 
	\EndIf
	\EndIf
	\EndIf
	\EndWhile
	\State return {\tru;}
	\EndProcedure\label{remvend}
	\algstore{remv}
\end{algorithmic}
\hrule
	\begin{algorithmic}[1]
\renewcommand{\algorithmicprocedure}{\textbf{Operation}}	
	\algrestore{remv}
	\small
	\Procedure{ $\conv$($k$)}{}\label{convstart}
	\State {$cv$ $\gets$ $\vh.\vnext$;}
	\While{($cv.k$ $<$ $k$)}
	\State{$cv$ $\gets$ \unMarkedRef($cv.\vnext$);}
	\EndWhile
	\If{($cv.k$ $=$ $k$ $\bigwedge$ $\neg$ \isMarked($cv$))} 
	\State return  {\tru;}\label{lin:conv-check}
	\Else {\hspace{2mm}return  {\fal;}}
	\EndIf
	\EndProcedure\label{convend}
	\algstore{conv}
\end{algorithmic}

	\end{subfigure}
    \begin{subfigure}{.3\textwidth}
	\begin{algorithmic}[1]
	\algrestore{conv}
\renewcommand{\algorithmicprocedure}{\textbf{Operation}}	
	\small
		\Procedure{ $\adde$($k$, $l$)}{}\label{addestart}
		\State {$\langle$ u, v, st $\rangle$ $\gets$ $\convplus$($k, l);$}  \label{lin:adde-convplus}
		 \If{(st = \fal)}
        \State {return ``\vntp'' ; }
        \EndIf
        \While{(\tru)}
	 \If{(\isMarked($u$) $\bigvee$  \isMarked($v$))}	\label{lin:adde-check-uv}
        \State return ``\vntp''; 
        \EndIf
		\State{$\langle pe, ce \rangle$ $\gets$ \loceplus ($u.\enext$, $l$);} \label{lin:adde-loceplus}
		\If{($ce.l$ $=$ $l$)}
		\State return ``\ep''; 
		\EndIf

        \State{$ne$ $\gets$ \createe($l$);}
		\State{$ne.\enext$ $\gets$ $ce$;}
		\State{$ne.\pointv$ $\gets$ $v$;}
		\If{(CAS($pe.\enext$, $ce, ne   $))}\label{lin:cas-adde}
         \State{u.\ecount.\fadd(1);}\label{faa:ade}
		\State return ``\eadd'';
		\EndIf
		
		\EndWhile
        \EndProcedure\label{addeend}
		\algstore{adde}
\end{algorithmic}
\hrule
\begin{algorithmic}[1]
	\algrestore{adde}
	\small
\renewcommand{\algorithmicprocedure}{\textbf{Operation}}	
	\Procedure{ $\cone$($k, l$)}{}\label{conestart}
    	\State {$\langle$ u, v, st $\rangle$ $\gets$ $\concplus$($k, l);$}  \label{lin:convplus-uv}
		 \If{(st = \fal)}
        \State {return ``\vntp''; }
        \EndIf
        \State {$ce$ $\gets$ $u.\enext$;}
		\While{($ce.l$ $<$ $l$)}
		\State{$ce$ $\gets$ \unMarkedRef($ce.\enext$);}
	    \EndWhile
	    \If{($ce.l$ $=$ $l$ $\bigwedge$ $\neg$ \isMarked($u$) $\bigwedge$ $\neg$ \isMarked($v$) $\bigwedge$ $\neg$ \isMarked($ce$));} \label{lin:cone-check}
        \State {return ``\ef'' ; }
        \Else
        \State {return ``\ventp''; }
        \EndIf
	    \EndProcedure\label{coneend}
		\algstore{cone}
\end{algorithmic}
	\end{subfigure}
	 \begin{subfigure}{.352\textwidth}
	\begin{algorithmic}[1]
	\algrestore{cone}
	\small
\renewcommand{\algorithmicprocedure}{\textbf{Operation}}	
		\Procedure{ $\reme$($k$, $l$)}{}\label{remestart}
    		\State {$\langle$ u, v, st $\rangle$ $\gets$ $\convplus$($k, l$);}  \label{lin:reme-validate-u}
		 \If {(st = \fal)}
	    \State {return ``\vntp''; }
       \EndIf
       	\While{(\tru)}
       	        \If{(\isMarked($u$) $\bigvee$  \isMarked($v$))}	\label{lin:reme-check-uv}
        \State return ``\vntp''; 
        \EndIf
		\State{$\langle pe, ce \rangle$ $\gets$ \loceplus ($u.\enext$, $l$);} \label{lin:reme-loceplus}
		\If{($ce.l$ $\neq$ $l$)}
		\State return ``\entp''; 
		\EndIf
		\State{$cnt$ $\gets$ $ce.\enext$;}
		\If{($\neg$ \isMarked($cnt$))} 
		\If{(CAS($ce.\enext$, $cnt$, \MarkedRef($cnt$)))} \label{lin:cas-reme}
		\State{$u.\ecount.\fadd(1)$;} \label{increm}
		 \If{(CAS($pe.\enext, ce, cnt$))}  { $break$;}\label{lin:cas-reme-phy}
                \EndIf
		\EndIf
		\EndIf
		\EndWhile
		\State return ``\er''; 
	    \EndProcedure\label{remeend}
        \algstore{reme}
\end{algorithmic}
\hrule
\begin{algorithmic}[1]
	\algrestore{reme}
	\small
		\Procedure{$\locvplus$($v$, $k$)}{}\label{locvstart}
		\While{($\tru$)} \label{lin:search-again}
		\State {$pv \gets v$; $cv \gets pv.\vnext$;} \label{lin:locv3-w}
		\While{$(\tru)$}
		\State{$cn \gets cv.\vnext$;}
	    \While{(\isMarked($cn)) \bigwedge (cv.k < k))$} \label{lin:locv5-w} 
		\If{($\neg CAS( pv.\vnext, cv, cv.\vnext)$)}
		\State goto \ref{lin:search-again}; 
		\EndIf
        \State {$cv \gets cn$; $cn \gets \unMarkedRef(cv.\vnext)$;} 
        \EndWhile
        \If{($cv.k \geq k$)} {return $\langle pv, cv \rangle$;}
        \EndIf
        \State {$pv \gets cv$; $cv \gets cn$;} 
		\EndWhile
        \EndWhile
		\EndProcedure\label{locvend}
		\algstore{locvplus}
\end{algorithmic}
	\end{subfigure}
    \vspace{-2mm}
	\caption{Pseudo-codes of \addv, \remv, \conv, \adde, \reme, \cone and \locvplus.}\label{fig:v-methods}
\end{figure*}

A \FAA${(address, val)}$ atomically increments the value at the memory location ${address}$ by $val$. A \texttt{CAS}${(address, old, new)}$ instruction checks if the current value at a memory location ${address}$ is equivalent to the given value ${old}$, and only if true, changes the value of ${address}$ to the new value ${new}$ and returns \tru; otherwise the memory location remains unchanged and the instruction returns \fal. Such a system can be perfectly realized by a Non-Uniform Memory Access (NUMA) computer with one or more multi-processor CPUs.

\subsection{The design basics}
The basic structure of the presented graph data structure is based on a linked-list. Therefore, for the lock-free synchronization in the graph, we utilize the approach of an existing lock-free linked-list algorithm \cite{Harrisdisc01}. The core idea of the design is a remove procedure  based on a protocol of first atomically injecting an \textit{operation descriptor} on the outgoing pointers of the \vnode{s} or \enode{s}, which are to be removed, and then atomically modifying the incoming pointers to disconnect the nodes from the \vlist or \elists. If multiple concurrent operations try to modify a pointer simultaneously, they synchronize by helping the pending removal operation that would have successfully injected its descriptor.

More specifically, to remove a \vnode (respectively \enode) $n$ from the \vlist (respectively an \elist), we use a \CAS to inject an operation descriptor at the pointer \vnext (respectively \enext). We call these descriptors a \textit{mark} and a pointer with a descriptor as \textit{marked}. We call a \vnode (respectively \enode) \textit{marked} if it \vnext (respectively \enext) pointer is marked. A pointer once marked is never modified again. 

A concurrent operation, if obstructed at a marked pointer, helps by performing the remaining step of a removal: modifying the incoming pointer from the previous node to point the next node in the list and thereby removing the node. An addition operation uses a single \CAS to update the target pointer only if it is not marked, called \textit{clean}, otherwise it helps the pending removal operation. During lookup or a \rbty query, a traversal on the \vlist or the \elist{s} does not perform any help. Traversals for modification operations help pending removal operations. After helping a pending removal operation, a concurrent addition or removal operation restart suitably.

To realize the atomic step to inject an operation descriptor, we replace a pointer using a \CAS with a single-word-sized packet of itself and an operation descriptor. In C/C++, to pack the operation descriptor with a pointer in a single memory-word, we apply the so-called \textit{bit-stealing}. In a x86/64 machine, where memory allocation is aligned on a 64-bit boundary, three least significant bits in a pointer are unused. The mark descriptor uses the last significant bit: if the bit is set the pointer is marked, otherwise clean.

For ease of exposition, we assume that a memory allocator always allocates a variable at a new address and thus an ABA problem does not occur. ABA is an acronym to indicate a typical problem in a CAS-based lock-free algorithm: a value at a shared variable can change from A to B and then back to A, which can corrupt the semantics of the algorithm. We assume the availability of a lock-free memory reclamation scheme.



\textbf{Pseudo-code convention:} The algorithm is presented in the pseudo-codes in the \Cref{fig:v-methods,fig:e2-methods,fig:getpath-methods}. If $x$ is a pointer pointing to a class instance, we use $x.y$ to indicate the field $y$ of the instance of the class. $<r_1, r_2,\ldots,r_n>$  indicates a return of multiple variables together. For a pointer $x$, $\MarkedRef(x)$ and $\unMarkedRef(x)$ denote $x$ with its last significant bit set to 1 and 0, respectively. $\isMarked(x)$ returns \tru if the last significant bit of $x$ is set to 1, otherwise, it returns \fal. A call of $\createv(k)$ instantiates a new \vnode with key $k$, whereas, $\createe(l)$ instantiates a new \enode with key $l$. For a newly instantiated \vnode or \enode, the pointer fields are \nul, integer array has 0 in each slot and an integer counter is 0.
\subsection{The \lf vertex operations}
\label{sec:working-con-graph-methods}
The operations \addv, \remv and \conv are shown in the \Cref{fig:v-methods}. Fundamentally, these operations in our algorithm are similar to the ones in \cite{Harrisdisc01}. However, unlike \cite{Harrisdisc01}, \conv does not indulge in helping during traversal. This essentially makes the \conv operations wait-free in case the set of keys is finite. For the sake of efficiency. We believe that \conv operation will be lot more frequent than update operations. Hence, we are not making the \conv operation help the update operations. Inspite of this, the algorithms are still lock-free. The return of the operations are as described in their respective sequential specifications presented in the \Cref{subsec:ds}.

An \addv($k$) operation, in the \cref{advstart} to \ref{advend}, first calls \locvplus procedure to locate the appropriate \vnode, ahead of which it needs to add the new \vnode. On locating a clean \vnext pointer of a \vnode with the key just less than $k$, it attempts a \CAS to add the new \vnode. On a \CAS failure, the process is reattempted. A \remv($k$), \cref{remvstart} to \ref{remvend}, similarly traverses down the \vlist by calling \locvplus. On locating the \vnode $v(k)$ to remove, it (a) marks the \vnext of $v(k)$ using a \CAS, and (b) atomically updates the \vnext of the previous node in the \vlist to point to the next \vnode after $v(k)$ using a \CAS. On any \CAS failure the process is reattempted.

During traversal in a \locvplus, \cref{locvstart} to \ref{locvend}, we help a pending \remv by essentially completing the step (b) as described above. 

A successful \CAS at \cref{lin:cas-remv} is called the \textit{logical removal} of $v(k)$. After this step any call of \isMarked($v(k).\vnext$) would return \fal, which is used by a \conv, \cref{convstart} to \cref{convend}, to return \fal in case a marked $v(k)$ is located. 

The removal of a vertex from a graph also requires removing all the incoming and outgoing edges of it. The outgoing edges are removed along with the \vnode $v(k)$, as it is logically, and eventually \textit{physically} detached from the \vlist by a successful \CAS at \cref{lin:cas-remv-phy}. However, the incoming edges are logically removed as any \pointv from an \enode of any \elist would call \isMarked to check the removed \vnode. As an optimization, we leave the \enode{s} in all the \elists with their \pointv pointing to $v(k)$ as they were. Eventually those \enode{s} are removed as part of the physical removal of the roots of their respective \elists. Note that, the physical removal can be performed by any helping operation. 

\subsection{The \lf edge operations}
 An $\adde(k, l)$ operation, \cref{addestart} to \ref{addeend}, starts by verifying the presence of vertices $v(k)$ and $v(l)$ in the \vlist of the graph by invoking the \convplus at the \cref{lin:adde-convplus}. The procedure \convplus, \cref{convpstart} to \cref{convpend}, essentially locates the \vnode with the smaller key first starting from $v(-\infty)$ and if located, starts from that \vnode to locate the \vnode with the bigger key. If any of the nodes not located $\adde(k, l)$ returns the string \vntp. Once both the nodes located, it adds a new \enode $e(k, l)$ in the \elist of $v(k)$ along the same lines of addition of a \vnode in the \vlist. To traverse down an \elist, the procedure \loceplus is called, see \cref{locestart} to \ref{loceend} in the \Cref{fig:e2-methods}. 
 
 A \loceplus traverses down an \elist and physically removes two kind of logically removed \enode{s}: (a) the ones corresponding to a logically removed \vnode, see the \cref{lin:loce-log-mark,lin:loce-phy}, and (b) logically removed \enode{s}, see the \cref{remlogd}. It returns the address of two consecutive \enode{s} between which the new \enode could be added.
 \begin{figure*}[!htp]
\captionsetup{font=footnotesize}
	\begin{subfigure}{.4\textwidth}	
\begin{algorithmic}[1]
	\algrestore{locvplus}
	\small
		\Procedure{$\loceplus$($v$, $k$)}{}\label{locestart}
		\While{(\tru)}  \label{lin:retry-locte}
		\State {$pe \gets v$; $ce \gets pe.\enext$;} 
        \While{(\tru)}
        \State{$cnt \gets ce.\enext$; \vnode  vn $\gets$ ce.\pointv;} 
        \While{(\isMarked($vn$) $\bigwedge$ $\neg$ \isMarked($cnt$))} \label{lin:retry-locte2}
		\If{($\neg$CAS($ce.\enext$, $cnt$, \MarkedRef($cnt$)))} \label{lin:loce-log-mark}
		\State{goto Line \ref{lin:retry-locte};}
		\EndIf
		\If{($\neg$CAS($pe.\enext, ce, cnt$))}\label{lin:loce-phy} {goto Line \ref{lin:retry-locte};}  
		\EndIf
        \State {$ce \gets cnt$; $ vn \gets ce.\pointv;$} 
        \State {$cnt \gets \unMarkedRef(ce.\enext)$;}
		\EndWhile
        \While{(\isMarked(cnt))}
		\State{v.\ecount.\fadd(1);}\label{faa:loc}
		\If{($\neg$ CAS($pe.\enext, ce, cnt))$;}\label{remlogd} {goto \ref{lin:retry-locte}; }
		\EndIf
        \State {$ce \gets cnt$;  $ vn \gets ce.\pointv;$} 
        \State {$cnt \gets \unMarkedRef(ce.\enext)$;}
        \EndWhile
        \If{(\isMarked($vn$))} { goto Line \ref{lin:retry-locte2};}
        \EndIf
        \If{($ce.l \geq k$)} {return $\langle pe, ce \rangle$} 
        
        \EndIf
        \State {$pe \gets ce$; $ce \gets cnt$;} 
		\EndWhile
		\EndWhile
        \EndProcedure\label{loceend}
		\algstore{locte}
\end{algorithmic}
	\end{subfigure}
    \begin{subfigure}{.32\textwidth}
\begin{algorithmic}[1]
	\algrestore{locte}
	\small
		\Procedure{$\convplus$ ($k$, $l$)}{}\label{convpstart}
		\If{($k < l$)}
		\State {$\langle pv1, cv1\rangle$ $\gets$ $\locvplus$($\vh$, $k$);} \label{lin:locvplus-k}
		\If{($cv1.k$ $\neq$ $k$)}
	     \State{return $\langle \nul, \nul, \fal \rangle$;}
		\EndIf
		\State {$\langle pv2, cv2\rangle$ $\gets$ $\locvplus$($cv1$, $l$);} \label{lin:locvplus-l}
		\If{($cv2.k$ $\neq$ $l$)}
	     \State{return $\langle \nul, \nul, \fal \rangle$;}
		\EndIf
		\Else
		\State {$\langle pv2, cv2\rangle$ $\gets$ $\locvplus$($\vh$, $l$);}
		\If{($cv2.k$ $\neq$ $l$)}
	     \State{return $\langle \nul, \nul, \fal \rangle$;}
		\EndIf
		\State {$\langle pv1, cv1\rangle$ $\gets$ $\locvplus$($cv2$, $k$);}
		\If{($cv1.k$ $\neq$ $k$)}
	     \State{return $\langle \nul, \nul, \fal \rangle$ ;} 
		\EndIf
		\EndIf
		\State {returns $\langle cv1, cv2, \tru \rangle$;} 
        \EndProcedure\label{convpend}
		\algstore{convplus}
\end{algorithmic}	
\hrule
\begin{algorithmic}[1]
	\algrestore{convplus}
	\small
		\Procedure{$\loccplus$($v$, $k$)}{}
	\State {$pv \gets v$; $cv \gets pv.\vnext$;} \label{lin:locv3-start}
		\While{($\tru$)} \label{lin:loccplus-search-again}
		\If{($cv.k \geq k$)}
        \State{return $\langle pv, cv \rangle$;}
        \EndIf
        \State {$pv \gets cv$; $cv \gets \unMarkedRef(cv.\vnext)$;} 
		\EndWhile
		\EndProcedure
		\algstore{loccplus}
	\end{algorithmic}
	\end{subfigure}
	 \begin{subfigure}{.23\textwidth}
	\begin{algorithmic}[1]
	\algrestore{loccplus}
	\small
		\Procedure{$\concplus$ ($k$, $l$)}{}
		\If{($k < l$)}
		\State {$\langle pv1, cv1\rangle$ $\gets$ $\loccplus$($\vh$, $k$);}
		\If{($cv1.k$ $\neq$ $k$)}
	     \State{return $\langle \nul, \nul, \fal \rangle$;}
		\EndIf
		\State {$\langle pv2, cv2\rangle$ $\gets$ $\loccplus$($cv1$, $l$);}
		\If{($cv2.k$ $\neq$ $l$)}
	     \State{return $\langle \nul, \nul, \fal \rangle$;}
		\EndIf
		\Else
		\State {$\langle pv2, cv2\rangle$ $\gets$ $\loccplus$($\vh$, $l$);}
		\If{($cv2.k$ $\neq$ $l$)}
	     \State{return $\langle \nul, \nul, \fal \rangle$;}
		\EndIf
		\State {$\langle pv1, cv1\rangle$ $\gets$ $\loccplus$($cv2$, $k$);}
		\If{($cv1.k$ $\neq$ $k$)}
	     \State{return $\langle \nul, \nul, \fal \rangle$ ;} 
		\EndIf
		\EndIf
		\State {returns $\langle cv1, cv2, \tru \rangle$;} 
        \EndProcedure
        \algstore{concplus}
\end{algorithmic}
	\end{subfigure}
    \vspace{-2mm}
	\caption{Pseudo-codes of \loceplus, \concplus, \loccplus and \concplus.}\label{fig:e2-methods}
\end{figure*}

 Before every attempt of executing a \CAS to add an \enode, an $\adde(k, l)$ checks if the \vnode{s} $v(k)$ and $v(l)$ are logically removed. This check ensures avoiding an interesting wrong execution as illustrated in the \Cref{fig:noseqhist}.  If the  \enode $e(k, l)$ is found in the \elist of $v(k)$, the string \ep is returned. A successful \CAS to add the \enode returns \eadd.

A $\reme(k, l)$ operation, \cref{remestart} to \cref{remeend}, works along the similar lines as $\adde(k, l)$ during traversal to locate the \vnodes $v(k)$ and $v(l)$ and the \enode $e(k, l)$. If all of them found in the data structure, it uses \CAS to first logically remove the \enode, \cref{lin:cas-reme}, and thereafter physically remove the \enode, \cref{lin:cas-reme-phy}. In case either of the vertices $v(k)$ or $v(l)$ not present in the data structure, it returns \vntp. If the \enode $e(k, l)$ not found, it returns \entp. On a successful \CAS to logically remove an \enode, it returns \er.

The operations $\adde(k, l)$ and $\reme(k, l)$ also increment an atomic counter at the node $v(k)$ using a \FAA, see the \cref{faa:ade,increm}.  This atomic counter facilitates in comparing the output of consecutive BFS traversals in the \getpath operation. We discuss it in the next section. A \loceplus, as it can be called from an \adde or a \reme, also increments the atomic counter before physically removing an \enode, see \cref{faa:loc}. It ensures that if a thread after logically removing an \enode got delayed then a helper thread increments the counter on its behalf.

A $\cone(k,l)$ operation, locates the \vnode{s} $v(k)$, $v(l)$ and the \enode{s} $e(k,l)$ similar to other edge operations. Before returning an indicative string as appropriate, it ensures that non of the nodes $v(k)$, $v(l)$ and $e(k,l)$ are marked, see the lines \ref{conestart} to \ref{coneend}. For the sake of efficiency. Like \conv, the \cone also does not indulge in helping the update operations during traversal.

\begin{figure}
\captionsetup{font=footnotesize}
	\centerline{\scalebox{0.6}{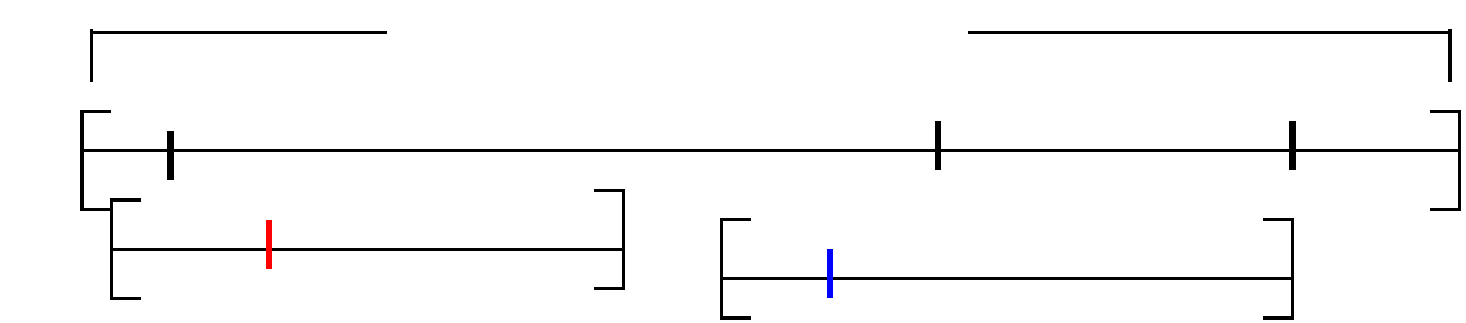}}
	\caption{The reason we need to check (Line \ref{lin:adde-check-uv}) $v(k)$ or $v(l)$ is marked or not in \adde method. Thread $T_1$ trying to perform $\adde(k, l)$, first invokes $\loccplus(k)$. After $T_1$ has verified $v(k)$ to be present in \vlist, thread $T_2$ deletes $v(k)$ and thread $T_3$ adds $v(l)$. In the absence of Line \ref{lin:adde-check-uv} check, $T_1$ confirms $v(l)$ to be present and adds an edge $(k,l)$, thereby returning \eadd. However, this is an \texttt{illegal execution}. In no possible sequential history equivalent to the given concurrent execution will both the vertices  $v(k)$ and $v(l)$ exist together. By having an check in this scenario, $\adde(k, l)$ will return \vntp on checking that vertex $v(k)$ has been deleted.}
	\label{fig:noseqhist}
\end{figure}

\begin{figure*}[!htp]
\captionsetup{font=footnotesize}
	\begin{subfigure}{.33\textwidth}	
\begin{algorithmic}[1]
	\algrestore{concplus}
		\small
\renewcommand{\algorithmicprocedure}{\textbf{Operation}}	
		\Procedure{\getpath($k$, $l$)}{}\label{getpstart}
        \State{ tid $\gets$ this\_thread.get\_id();} 
	    \State {$\langle$ u, v, st $\rangle$ $\gets$ $\concplus$($k, l);$}  \label{lin:getpath-concplus}
		 \If{(st = \fal)}
	     \State {return \nul; } 
        \EndIf
        	\If{(\isMarked($u$) $\bigvee$  \isMarked($v$))}
        \State return {\nul;} 
        \EndIf
         \State{list $<\bfsnode>$  bTree; } 
         \State{$\langle$ st, bTree$\rangle$ $\gets$ \scan($u$, $v$, $tid$);} \label{lin:getpath-scan}
         \If{(st = \tru)}
        \State {return ``PATH $u$ TO $v$ FROM THE bTree''; }
         \Else {\hspace{2mm}return \nul; }
         \EndIf 
		\EndProcedure\label{getpend}
        \algstore{getpath}
\end{algorithmic}	
    \hrule
	\begin{algorithmic}[1]
	\algrestore{getpath}
		\small
		\Procedure{\scan($u$, $v$, $tid$)}{}\label{scanstart}
         \State{list $<\bfsnode>$  ot, nt ; } 
          \State{bool f1 $\gets$ \treec($u$, $v$, $ot$, $tid$); }
        \While{(\tru)} \label{whilescan}
           \State{bool f2 $\gets$ \treec($u$, $v$, $nt$, $tid$); } 
           \If{(f1 = \tru $\bigwedge$ f2 = \tru $\bigwedge$ \comparepath(ot, nt))}  \label{lin:scan-comparepath}
           \State{return $\langle$f2, nt$\rangle$;} 
           \Else 
           \If{(f1 = \fal $\bigwedge$ f2 = \fal $\bigwedge$ \comparetree(ot, nt))}\label{lin:scan-comparetree}
           \State{return $\langle$f2, \nul$\rangle$;}
           \EndIf
           \EndIf
           \State{f1 $\gets$ f2; ot $\gets$ nt;}\label{rescan}
    \EndWhile
		\EndProcedure\label{scanend}
        \algstore{scan}
\end{algorithmic}
	\end{subfigure}
    \begin{subfigure}{.33\textwidth}
	\begin{algorithmic}[1]
	\algrestore{scan}
		\small
		\Procedure{ \treec($u$, $v$, $bTree$, $tid$)}{}\label{trecstart}
       \State{queue $<$\bfsnode$>$ que;} 
         \State{cnt $\gets$ cnt $+$ 1; } 
         \State{u.visitedArray[tid] $\gets$ cnt ;} 
        \State{\bfsnode bNode(u, 0, \nul, \nul);}
        \State{bTree.Insert(bNode);} 
        \State{que.enque(bNode);} 
        \While{($\neg$que.empty())} 
        \State{\bfsnode  cvn $\gets$ que.deque();} 
        \State{eh $\gets$ cvn.n.enext;} 
        \For{(\enode itn $\gets$ eh.\enext to itn.\enext $\neq$ \nul)} 
       
        \If{($\neg$\isMarked(itn))}  
        \State{\vnode adjn $\gets$ itn.\pointv;}
        \If{($\neg$\isMarked(adjn))} 
        \If{(adjn = v)}
        \State{\bfsnode bNode(adjn, adjn.\lecount, cvn, \nul);}
         \State{bTree.Insert(bNode);} 
        \State{return \tru;} 
        \EndIf
        \If{($\neg$ \checkvisited(tid, adjn, cnt))}
        \State{adjn.\visitedarray[tid] $\gets$ cnt ;} 
        \State{\bfsnode bNode(adjn, adjn.\lecount, cvn, \nul);} 
         \State{bTree.Insert(bNode);} 
        \State{que.enque(bNode);} 
        
        \EndIf
        \EndIf
        \EndIf
        \EndFor
    \EndWhile
     \State{return \fal;}
		\EndProcedure\label{trecend}
        \algstore{getpathc}
\end{algorithmic}
	\end{subfigure}
	 \begin{subfigure}{.33\textwidth}
	\begin{algorithmic}[1]
	\algrestore{getpathc}
		\small
		\Procedure{ \comparetree($ot, nt$)}{}\label{ctrstart}
	 \If{(ot = \nul $\bigvee$ nt = \nul)}
	     \State {return $\fal$ ; }
        \EndIf
        \State{\bfsnode oit $\gets$ ot.Head, nit $\gets$ nt.Head;} 
        \Do 
        \If{(oit.n $\neq$ nit.n $\bigvee$ oit.\lecount $\neq$ nit.\lecount $\bigvee$ oldit.p $\neq$ newit.p)} {\hspace{2mm}return \fal; }
        \EndIf
        \State{oit $\gets$  oit.\bnext; nit $\gets$  nit.\bnext;}
        \doWhile{(oit $\neq$ ot.Tail $\bigwedge$ nit $\neq$ nt.Tail );} 
         \If{(oit.n $\neq$ nit.n $\bigvee$ oit.\lecount $\neq$ nit.\lecount $\bigvee$ oit.p $\neq$ nit.p)} {\hspace{2mm}return \fal; }
         \Else {\hspace{2mm}return $\tru$ ; }
         \EndIf
		\EndProcedure\label{ctrend}
        \algstore{comparetree}
\end{algorithmic}
\hrule
\begin{algorithmic}[1]
	\algrestore{comparetree}
		\small
		\Procedure{ \comparepath($ot, nt$)}{}\label{compathstart}
	 \If{(ot = \nul $\bigvee$ nt = \nul)}
	     \State {return $\fal$ ; }
        \EndIf
        \State{\bfsnode oit $\gets$ ot.Tail, nit $\gets$ nt.Tail;}
        \Do 
         \If{(oit.n $\neq$ nit.n $\bigvee$ oit.\lecount $\neq$ nit.\lecount $\bigvee$ oldit.p $\neq$ newit.p)} {\hspace{2mm}return \fal; }
        \EndIf
        \State{oit $\gets$  oit.p; nit $\gets$  nit.p;}
        \doWhile{(oit $\neq$ ot.Head $\bigwedge$ nit $\neq$ nt.Head );} 
         \If{(oit.n $\neq$ nit.n $\bigvee$ oit.\lecount $\neq$ nit.\lecount $\bigvee$ oit.p $\neq$ nit.p)} {\hspace{2mm}return \fal; }
         
         \Else {\hspace{2mm}return \tru; }
         \EndIf
		\EndProcedure\label{compathend}
	\end{algorithmic}
	\end{subfigure}
    \vspace{-2mm}
	\caption{Pseudo-codes of \getpath, \scan, \treec, \comparetree and \comparepath.}\label{fig:getpath-methods}
\end{figure*}

\subsection{The \of \getpath operation}

The design of the \getpath(see \Cref{fig:getpath-methods}) operation draws from the snapshot algorithm proposed by Afek et al. \cite{Afek}. A $\getpath(k, l)$, \cref{getpstart} to \ref{getpend}, first checks the presence of $v(k)$ and $v(l)$ by invoking the \concplus procedure at the \cref{lin:getpath-concplus}. After successfully checking the presence of both the vertices it goes to perform repeated BFS traversals by invoking the procedure \scan at the \cref{lin:getpath-scan}. If the \vnodes $v(k)$, $v(l)$ are not located, it right away returns \nul. 

The \scan procedure, \cref{scanstart} to \ref{scanend}, first initializes two lists of \bfsnode{s}. We call such a list a \bfstree. The two \bfstree{s} are used to hold collection of \vnode{s} in two consecutive BFS traversals. A reader not familiar with the BFS traversals in graphs may refer to any standard book on algorithms such as \cite{cormen2009introduction}. The procedure \treec, \cref{trecstart} to \ref{trecend}, takes in a \bfstree and fills it with nodes collected in a BFS traversal. A BFS traversal terminates as soon as the \vnode $v(l)$ is located. However, in case $v(l)$ could not be reached from $v(k)$, the traversal terminates after exhausting all the outgoing edges from $v(k)$ represented by the \enode{s} in its \elist. 

During a BFS traversal, we put markers on the \vnode{s} to keep track of the visited ones, see \cite{cormen2009introduction}. In a sequential implementation, a single boolean variable is good enough. However, in our case, when multiple threads perform BFS traversals not only concurrently but also repeatedly, a single boolean variable or for that matter an boolean array would not suffice. To keep track of visited \vnodes, we use the array \visitedarray in them. The size of \visitedarray is equal to the number of threads in the shared-memory system. Thus, a slot of \visitedarray, used as a counter for the number of visits, provides local marker for repeated traversals by a thread.

The return of the \treec procedure is a boolean indicating if $v(l)$ was located. If the return of two consecutive \treec do not match we discard the old \bfstree and start collecting a new one, see the \cref{rescan}. However, if the returns match and both are \tru, it indicates that both the times a path from $v(k)$ to $v(l)$ could be discovered. Hence, we compare the collected paths, which are subset of the two \bfstree{s}. Please notice that we can not return either of them unless the two paths match because we are not sure if either existed at any instant during the lifetime of \getpath. 


The procedure \comparepath, \cref{compathstart} to \ref{compathend}, compares two \bfstree{s} with respect to the paths between $v(k)$ and $v(l)$: it starts from the last \bfsnode{s} in the two \bfstree{s} and follows the \bfsnode-pointers $p$ that takes to the previous node in a possible path; at any \bfsnode if there is a mismatch between the contained \vnodes, it terminates. 

If the returns of two consecutive \treec are \fal, it indicates that both the times a path was not traced between $v(k)$ and $v(l)$. However, to be sure that during the entire lifetime of the $\getpath(k, l)$, at every point in time no path ever existed, we need to compare the two returned \bfstree{s}. The comparison of the two \bfstree{s} in entirety is done in the procedure \comparetree, see the \cref{ctrstart} to \ref{ctrend}. 

If the comparison of two consecutive \bfstree{s} do not match in procedures \comparetree or \comparepath, we discard the first \bfstree and restart the tree-collection.

While comparing two \bfstree{s} or the paths therein in the procedures \comparetree or \comparepath, the \bfstree{s} require to be compared with respect to not just the \bfsnode{s}, but also the counters \lecount of the \bfsnode{s} contained in them. We explain this requirement below. 

Consider an adversary against a \getpath operation. Consider two consecutive BFS traversals. Suppose that during the first traversal after the \treec discovered that no path to $v(l)$ existed via a \vnode $v_{i}$ and therefore moved to another \vnode $v_{i+1}$ and continued the traversal until its exploration exhausted. As we know that once $v_{i}$ is visited, it will not be revisited. Now suppose that when \treec was visiting nodes after $v_{i+1}$, the adversary added an edge $(v_{i},v(l))$ that made a path exist between $v(k)$ and $v(l)$. However, before we could start the second BFS traversal, the edge $(v_{i},v(l))$ was removed by the adversary bringing the graph exactly at the same state at which the first traversal had started. Now, suppose that even during the second traversal the same game is played by the adversary. In such a scenario, if we just matched the two consecutive \bfstree{s} with respect the the collected \bfsnode{s}, the operation $\getpath(k, l)$ would return \nul indicating that a path did not exist during its lifetime, which would be clearly incorrect.

During the edge modification operations, before an \enode $(v_{i},v(l))$ is physically removed, the atomic counter \ecount at $v_{i}$ is necessarily incremented by either the operation that logically removed $(v_{i},v(l))$ or a helping operation at the lines \ref{increm} or \ref{faa:ade} or \ref{faa:loc}. This ensures that we get to check the adversaries as described above. Although, it may make a \getpath continue until all the modification operations in the graph stop, we still have an ensured correct return of a \rbty query.

\label{sec:getPath}

\section{Correctness: Linearization Points}\label{sec:proof}
\subsection{Linearizability}
\begin{figure*}
	\captionsetup{font=footnotesize}
	\resizebox{0.97\linewidth}{!}{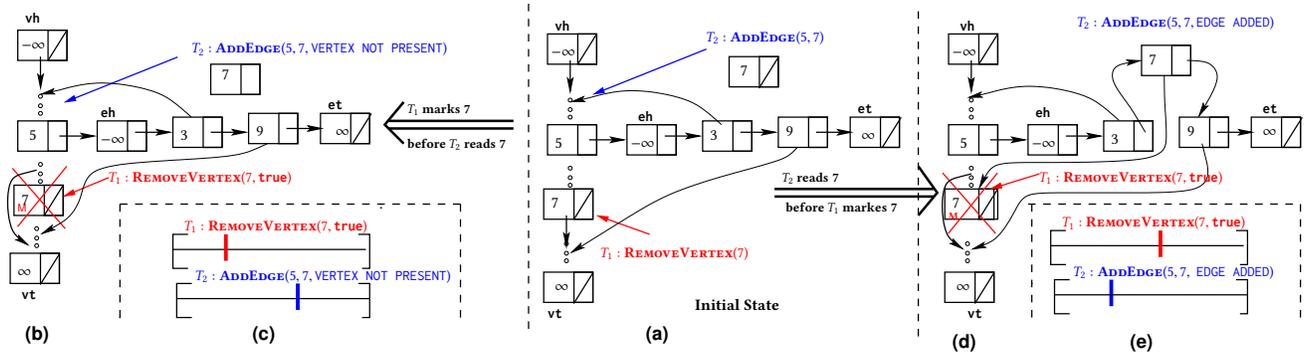}
	\caption{Execution with two concurrent operations - $\adde$ and $\remv$}
	\label{fig:lp-figs}
\end{figure*}
The ADT operations implemented by the data structure are represented by their invocation and return steps. To prove the correctness of an arbitrary concurrent execution of a set of ADT operations, we show that they satisfy the consistency framework \textit{\lbty} \cite{HerlWing:1990:TPLS}. To do that, firstly we show that it is possible to assign an atomic step as a \emph{linearization point} (\emph{\lp}) inside the execution interval of each of the operations. Thereafter, we also show that the data structure invariants are maintained across the \lp{s}. Thereby, it proves that an arbitrary concurrent execution is equivalent to a valid sequential execution obtained by ordering the operations by their \lp{s}.
\begin{theorem}\normalfont The ADT operations implemented by the presented algorithm are linearizable.\end{theorem}
\begin{proof}
For ease of presentation, we discuss the LPs in case by case manner depending on the return of the operations. 
\begin{enumerate}
	\item $\addv(k)$: Two cases based on the return values:
	\begin{enumerate}
		\item \tru: Successful \cas execution at the \cref{lin:cas-addv}.
		\item \fal: The vertex is already present. The \lp is the atomic read of the \vnext pointer pointing to $v(k)$. 
	\end{enumerate}
	\item $\remv(k)$: Two cases based on the return values:	
	\begin{enumerate}
		\item \tru: Successful \CAS execution at the \cref{lin:cas-remv} (logical removal).
		\item \fal: If a concurrent \remv operation $op$ removed $v(k)$ then just after the LP of $op$. $v(k)$ did not exist in the \vlist then at the invocation of $\remv(k)$.\label{step:remv-fal} 
	\end{enumerate}
	\item $\conv(k)$: Two cases based on the return value:	
	\begin{enumerate}
		\item \tru: The atomic read of the \vnext pointing to $v(k)$.
		\item \fal: same as the case \ref{step:remv-fal}, where \remv returns \fal. 
	\end{enumerate}
	\item $\adde(k, l)$: Three cases depending on the return values:
	\begin{enumerate}
		\item \eadd: Two sub-cases depending on if there is a concurrent $\remv(k)$ or a $\remv(l)$: \label{step:eadd}
		\begin{enumerate}
			\item No concurrent $\remv(k)$ or $\remv(l)$: the successful \CAS execution at the \cref{lin:cas-adde}.
			\item With concurrent $\remv(k)$ or $\remv(l)$: just before the first remove's \lp. A sample case for determining the \lp when a \remv is concurrent to an \adde is shown in the \Cref{fig:lp-figs}.
		\end{enumerate}
		\item \ep: Similar to the case when \eadd is returned.
		\begin{enumerate}
			\item No concurrent $\remv(k)$ or $\remv(l)$ or $\reme(k, l)$: the atomic read of the \enext pointer pointing to $e(k, l)$.
			\item With concurrent $\remv(k)$ or $\remv(l)$ or $\reme(k, l)$: just before the first remove's (either vertex or edge) \lp.
		\end{enumerate}		
		\item \vntp: Two sub-cases
		\begin{enumerate}
			\item If both $v(k)$ and $v(l)$ were in the \vlist at the invocation of $\adde(k, l)$ and a concurrent \remv removed $v(k)$ or $v(l)$ or both then just after the LP of the earlier \remv.
			\item If $v(k)$ and/or $v(l)$ were not present at the invocation of $\adde(k, l)$ then the invocation point itself.
		\end{enumerate}
	\end{enumerate}			
	\item $\reme(k, l)$: Similar to \adde, we have three cases depending on the return values:
	\begin{enumerate}
		\item \er: This is similar to the case \ref{step:eadd} of \adde. We have two sub-cases depending on if there are any concurrent \remv \mth{s}:
		\begin{enumerate}
			\item No concurrent $\remv(k)$ or $\remv(l)$: the successful \CAS execution at the Line \ref{lin:cas-reme}.
			\item With concurrent $\remv(k)$ or $\remv(l)$: just before the first remove's \lp.
		\end{enumerate}
		\item \entp: If a concurrent \reme operation removed $e(k,l)$ then just after its \lp, otherwise at the invocation of $\reme(k, l)$ itself. 
		
		\item \vntp: Absolutely same as the case \adde{} returning ``\vntp''.
		
	\end{enumerate}			
	
	\item $\cone(k, l)$: Similar to \reme, we have three cases depending on the return values. All the steps are very similar \reme. 
	\begin{enumerate}
		\item \ep: We have two sub-cases depending on if there are any concurrent \remv \mth{s}:
		\begin{enumerate}
			\item No concurrent $\remv(k)$ or $\remv(l)$: the atomic read of the \enext pointer pointing to $e(k, l)$..
			\item With concurrent $\remv(k)$ or $\remv(l)$: just before the \lp of the earlier \remv. 
		\end{enumerate}
		\item \entp: Absolutely same as the case of \reme returning ``\entp''.
		\item \vntp: Absolutely same as the case of \adde returning ``\vntp''.
	\end{enumerate}	
	
	\item $\getpath(k, l):$	Here, there are two cases:
	\begin{enumerate}
		\item \getpath invoke the \scan \mth: Assuming that \scan invokes $m$ (greater than equal to 2) \treec procedures. Then it is the last atomic read step of the $(m-1)^{st}$ \treec call.
		\item \getpath does not invoke the \scan \mth: If a concurrent \remv operation $op$ removed $v(k)$ or $v(l)$. Then just after the LP of $op$. If $v(k)$ or $v(l)$ did not exist in the \vlist before the invocation then at the invocation of $\getpath(k, l)$.
\end{enumerate}
\end{enumerate}
In the above discussion it is easy to observe that each of the LPs belong the interval between the invocation and return steps of the corresponding operations.

We can also observe in the algorithm that in any call of an $\addv(k)$, the traversal terminates at the \vnode where the key is just less than or equal to $k$. Similar argument holds true for a call of $\adde(k,l)$. Before every reattempt of $\addv(k)$ and $\adde(k,l)$, a traversal is performed following the sorted order of the \vlist and the \elists. This ensures that addition of a new \vnode or \enode does not violate the invariants of the data structure. The removal of a \vnode or an \enode by a \remv or \reme operation by default does not disturb the sorted order of the \vlist or any \elist. The lookup and \getpath operations do not modify the data structure. Thus, it can be observed that the operations maintain the invariants of the data structure across their \lp{s}.

This completes the proof of \lbty.
\end{proof}
\ignore{
 \begin{lemma}
\label{lem:lin}
The history $H$ generated by the interleaving of any of the methods of the \cgds, is \lble. \\
\textnormal{Proof in \cite{PeriSS16}}.
\end{lemma}

\begin{lemma}
\label{lem:lf}
The methods of the \cgds have guarantees the \lf progress guarantees.\\
\textnormal{Proof in \cite{PeriSS16}}.
\end{lemma}
}
\subsection{Non-blocking Progress}
\begin{theorem}\normalfont
	In the presented algorithm
	\begin{enumerate}[label=(\roman*)]
		\item If the set of keys is finite, the operations \conv and \cone are wait-free.\label{lflem1}
		\item The operation \getpath is \of. \label{lflem2}
		\item The operations \addv, \remv, \conv, \adde, \reme, and \cone are \lf. \label{lflem3}
	\end{enumerate}
\end{theorem}
\begin{proof}
It is easy to show \ref{lflem1}. If the set of keys is finite, the graph size has a fixed upper bound. Which implies that there are only finite number of \vnodes between $v(-\infty)$ and $v(\infty)$ in the \vlist. A $\conv(k)$ necessarily terminates on reaching $v(\infty)$ which will be done in a finite number of steps of any non-faulty thread. A similar argument holds true for a \cone.
	
It is also easy to see in the algorithm that whenever a modification operation is concurrent to a \getpath, a  \comparepath or a \comparetree can not return true, enforcing the While loop at the \cref{whilescan} in \scan procedure to not terminate. Therefore, unless the steps are taken in isolation, a non-faulty thread that calls the \getpath will not return as long as a concurrent operation is an arbitrary ADT operation. This shows \ref{lflem2}. 
	 
	 In the design of our algorithm, we can see that whenever an addition or a removal operation is obstructed by a concurrent removal operation by way of a marked pointer, the obstructing removal operation is necessarily helped ensuring its return. Addition and lookup do not require help by a concurrent operation. Therefore in an arbitrary concurrent execution comprising of any arbitrary data structure operation, at least one operation would complete in a finite number of steps taken by a non-faulty thread. Thus, the operations \addv, \remv, \conv, \adde, \reme, and \cone are \lf. This shows \ref{lflem3}.

\end{proof}

\section{Simulation Results and Analysis}
\label{sec:results}

We performed our tests on a workstation with Intel(R) Xeon(R) E5-2690 v4 CPU containing 56 cores running at 2.60GHz. Each core supports 2 logical threads. Every core's L1 - 64K, L2 - 256K cache memory is private to that core; L3-35840K cache is shared across the cores. The tests were performed in a controlled environment, where we were the sole users of the system. The implementation\footnote{The complete source code of our implementation is available on Github \cite{Pdcrl-ConcurrentGraphDS2017}.} has been done in C/C++ (without any garbage collection) and multi-threaded implementation is based on Posix threads.

In the experiments, we start with an initial graph of $1000$ vertices and approximately $\binom{1000}{2}/4 \approx 125000$ edges added randomly. The number of edges is approximately a quarter of the total number of edges in a complete graph with $1000$ vertices. When the program begins, it creates a fixed number of threads (1, 10, 20, 30, 40, 50, 60 and 70) and each thread randomly performs a set of operations chosen by a particular workload distribution. The evaluation metric used is the number of operations completed in a unit time. We measure throughput obtained on running the experiment for 20 seconds. Each data point is obtained by averaging over 5 iterations. 
We compare the non-blocking graph with its sequential and coarse-grained counterparts in two separate sets of experiments comprising: (a) the ADT operations excluding \getpath, and (b) all the ADT operations. In the experiments, the following workload distributions were considered.

In the first set of experiments, the distribution over the ordered set of operations $\{$\addv, \remv, \conv, \adde, \reme, \cone$\}$ are (1) \textit{Lookup Intensive}: $($$2.5\%$, $2.5\%$, $45\%$, $2.5\%$, $2.5\%$, $45\%$$)$, see the \figref{cgds-90-10}. (2) \textit{Equal Lookup and Updates}: $($$12.5\%$, $12.5\%$, $25\%$, $12.5\%$, $12.5\%$, $25\%$$)$, see the \figref{cgds-50-50}. (3) \textit{Update Intensive}: $($$22.5\%$, $22.5\%$, $5\%$, $22.5\%$, $22.5\%$, $5\%$$)$, \figref{cgds-10-90}. 

In the second set of experiments, the distribution over the ordered set of operations $\{$\addv, \remv, \conv, \adde, \reme, \cone, \getpath$\}$ are -  (1) \textit{Lookup Intensive}: $($$2\%$, $2\%$, $45\%$, $2\%$, $2\%$, $45\%$, $2\%$$)$, see the  \figref{gds-90-10}. (2) \textit{Equal Lookup and Updates}: $($$24\%$, $24\%$, $12.5\%$, $24\%$, $24\%$, $12.5\%$, $2\%$$)$, see the  \figref{gds-50-50}. (3) \textit{Update Intensive}: $($$22.5\%$, $22.5\%$, $4\%$, $22.5\%$, $22.5\%$, $4\%$, $2\%$$)$, see the  \figref{gds-10-90}. In this set of experiments, please note that we always take only $2\%$ \getpath operations considering that its overhead in comparison to other operations is significant.

In the plots, firstly, we observe that the non-blocking algorithm is highly scalable with the number of threads in the system: only after the available cores saturate with the threads, which is at $56$ threads, the throughput numbers stop increasing. On the other hand, the coarse-grained lock-based version shows performance degradation with the increasing number of threads. This is in line with the serious drawback associated with coarse-grained locks: the contention among threads to acquire the global lock actually increases with the increasing number of threads. In fact, the coarse-grained lock-based implementation performs worse than even the sequential implementation, as soon as concurrency kicks in the system. In terms of absolute numbers, we see a marginal difference in the two sets of experiments. It indicates that the overhead of \rbty query is light. On an average, we observe that the non-blocking algorithm offers $5$-$7$x increase in the throughput in comparison to the sequential counterpart, which in any case outperforms the coarse-grained locking algorithm with multiple threads. 
\begin{figure}[H]
\captionsetup{font=footnotesize}
\begin{subfigure}[b]{0.23\textwidth}
    \captionsetup{font=footnotesize, justification=centering}
   \caption{Lookup Intensive}
    \centering
        \resizebox{\linewidth}{!}{
   	\begin{tikzpicture}[scale=0.230]
	\begin{axis}[
    ybar,
    xmin=1,
    xmax=70,
    enlargelimits=0.2,
   xticklabels={1, 10, 20, 30, 40, 50, 60, 70},
   xlabel near ticks,
   xtick=data,
    bar width=5pt,
    legend style={at={(0.21,0.97)},anchor=north},
	xlabel=No of threads,
	ylabel=Throughput ops/sec, 
    ylabel near ticks,]
	\addplot [butter1, fill=butter1] table [x=Threads, y=$Sequential$]{results/ds90101.dat};
		\addlegendentry{$Sequential$}
	\addplot[chameleon1,fill=chameleon1] table [x=Threads, y=$Coarse$]{results/ds90101.dat};
    	    \addlegendentry{$Coarse$}
	\addplot [plum1,fill=plum1] table [x=Threads, y=$Lock-Free$]{results/ds90101.dat};
    \addlegendentry{$Lock\text{-}Free$}
    \addplot[blue,sharp plot,update limits=false] 
	coordinates {(-15,37468) (90,37468)} 
	node[above] at (axis cs:1,37468) {base-line};
	\end{axis}
	\end{tikzpicture}
        }
        \label{fig:cgds-90-10}
    \end{subfigure}
    \begin{subfigure}[b]{0.23\textwidth}
    \setlength{\belowcaptionskip}{2mm}
    \captionsetup{font=footnotesize, justification=centering}
     \caption{Equal Lookup and Updates}
    \centering
    \resizebox{\linewidth}{!}{
	\begin{tikzpicture} [scale=0.230]
	\begin{axis}[
    ybar,
    xmin=1,
    xmax=70,
    enlargelimits=0.2,
   xticklabels={1, 10, 20, 30, 40, 50, 60, 70},
   xlabel near ticks,
   xtick=data,
    bar width=5pt,
    legend style={at={(0.5,121,0.97)},anchor=north},
	xlabel=No of threads,
	ylabel=Throughput ops/sec,
    ylabel near ticks]
	\addplot [butter1, fill=butter1] table [x=Threads, y=$Sequential$]{results/ds50501.dat};
	\addplot[chameleon1,fill=chameleon1] table [x=Threads, y=$Coarse$]{results/ds50501.dat};
	\addplot[plum1,fill=plum1] table [x=Threads, y=$Lock-Free$]{results/ds50501.dat};
    \addplot[blue,sharp plot,update limits=false] 
	coordinates {(-15,14023) (90,14023)} 
	node[above] at (axis cs:1,14023) {base-line};
	\end{axis}
	\end{tikzpicture}
        }
        \label{fig:cgds-50-50}
    \end{subfigure}

\vspace{-4mm}
\begin{subfigure}[b]{0.23\textwidth}
    \captionsetup{font=footnotesize, justification=centering}
   \caption{Update Intensive}
        \centering
        \resizebox{\linewidth}{!}{
           \begin{tikzpicture}[scale=0.230]
	\begin{axis}[
    ybar,
    xmin=1,
    xmax=70,
    enlargelimits=0.2,
   xticklabels={1, 10, 20, 30, 40, 50, 60, 70},
   xlabel near ticks,
   xtick=data,
    bar width=5pt,
    legend style={at={(0.621,0.6597)},anchor=north},
	xlabel=No of threads,
	ylabel=Throughput ops/sec,
    ylabel near ticks]
	\addplot [butter1, fill=butter1] table [x=Threads, y=$Sequential$]{results/ds10901.dat};
	\addplot[chameleon1,fill=chameleon1] table [x=Threads, y=$Coarse$]{results/ds10901.dat};
	\addplot[plum1,fill=plum1] table [x=Threads, y=$Lock-Free$]{results/ds10901.dat};
    \addplot[blue,sharp plot,update limits=false] 
	coordinates {(-15,9147) (90,9147)} 
	node[above] at (axis cs:1,9147) {base-line};
	\end{axis}
	\end{tikzpicture}
        }
        \label{fig:cgds-10-90}
    \end{subfigure}
\begin{subfigure}[b]{0.23\textwidth}
    \captionsetup{font=footnotesize, justification=centering}
    \caption{ Lookup Intensive}
    \centering
        \resizebox{\linewidth}{!}{
   	\begin{tikzpicture}[scale=0.230]
	\begin{axis}[
    ybar,
    xmin=1,
    xmax=70,
    enlargelimits=0.2,
   xticklabels={1, 10, 20, 30, 40, 50, 60, 70},
   xlabel near ticks,
   xtick=data,
    bar width=5pt,
    legend style={at={(0.721,0.697)},anchor=north},
	xlabel=No of threads,
	ylabel=Throughput ops/sec,
    ylabel near ticks]
	\addplot [butter1, fill=butter1] table [x=Threads, y=$Sequential$]{results/gds90101.dat};
	\addplot[chameleon1,fill=chameleon1] table [x=Threads, y=$Coarse$]{results/gds90101.dat};
	\addplot[plum1,fill=plum1] table [x=Threads, y=$Lock-Free$]{results/gds90101.dat};
    \addplot[blue,sharp plot,update limits=false] 
	coordinates {(-15,40033) (90,40033)} 
	node[above] at (axis cs:1,40033) {base-line};
	\end{axis}
	\end{tikzpicture}
        }
        \label{fig:gds-90-10}
    \end{subfigure}
    
    \begin{subfigure}[b]{0.23\textwidth}
    \captionsetup{font=footnotesize, justification=centering}
    \setlength{\belowcaptionskip}{2mm}
    \caption{Equal Lookup and Updates}
    \centering
    \resizebox{\linewidth}{!}{
	\begin{tikzpicture} [scale=0.230]
	\begin{axis}[
    ybar,
    xmin=1,
    xmax=70,
    enlargelimits=0.2,
   xticklabels={1, 10, 20, 30, 40, 50, 60, 70},
   xlabel near ticks,
   xtick=data,
    bar width=5pt,
    legend style={at={(0.5,121,0.97)},anchor=north},
	xlabel=No of threads,
	ylabel=Throughput ops/sec,
    ylabel near ticks]
	\addplot [butter1, fill=butter1] table [x=Threads, y=$Sequential$]{results/gds50501.dat};
	\addplot[chameleon1,fill=chameleon1] table [x=Threads, y=$Coarse$]{results/gds50501.dat};
	\addplot [plum1,fill=plum1] table [x=Threads, y=$Lock-Free$]{results/gds50501.dat};
    \addplot[blue,sharp plot,update limits=false] 
	coordinates {(-15,13917) (90,13917)} 
	node[above] at (axis cs:1,13917) {base-line};
	\end{axis}
	\end{tikzpicture}
        }
        \label{fig:gds-50-50}
    \end{subfigure}
    \begin{subfigure}[b]{0.23\textwidth}
    \captionsetup{font=footnotesize, justification=centering}
    \caption{ Update Intensive}
        \centering
        \resizebox{\linewidth}{!}{
           	\begin{tikzpicture} [scale=0.230]
	\begin{axis}[
    ybar,
    xmin=1,
    xmax=70,
    enlargelimits=0.2,
   xticklabels={1, 10, 20, 30, 40, 50, 60, 70},
   xlabel near ticks,
   xtick=data,
    bar width=5pt,
    legend style={at={(0.5,121,0.97)},anchor=north},
	xlabel=No of threads,
	ylabel=Throughput ops/sec,
    ylabel near ticks]
	\addplot [butter1, fill=butter1] table [x=Threads, y=$Sequential$]{results/gds10901.dat};
	\addplot[chameleon1,fill=chameleon1] table [x=Threads, y=$Coarse$]{results/gds10901.dat};
	\addplot[plum1,fill=plum1] table [x=Threads, y=$Lock-Free$]{results/gds10901.dat};
    \addplot[blue,sharp plot,update limits=false] 
	coordinates {(-15,9147) (90,9147)} 
	node[above] at (axis cs:1,9147) {base-line};
	\end{axis}
	\end{tikzpicture}
        }
        \label{fig:gds-10-90}
    \end{subfigure}

\caption{Concurrent Graph Data-Structure Results without and with \getpath. (a), (b) and (c) are without \getpath and (d), (e) and (f) are with\getpath.} 
\label{fig:getpath}
\end{figure}
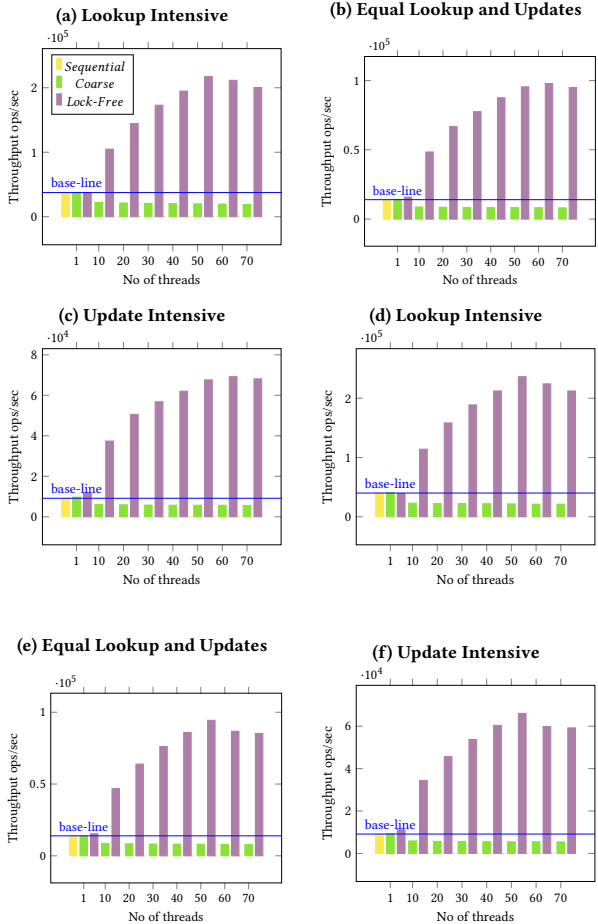

\section{Conclusion}\label{sec:conc}

In this paper, we presented a non-blocking algorithm to implement a dynamic \cgds, which allows threads to concurrently add and remove vertices/edges. The most important contribution of this work is an efficient obstruction-free reachability query in a concurrent graph. We prove the \lbty of the data structure operations. We extensively evaluated a sample C/C++ implementation of the algorithm through a number of micro-benchmarks. The non-blocking algorithm compared to a sequential and a coarse-grained lock-based concurrent version observably achieves up to $5$-$7$x speedup with respect to their throughput.

There are several graph databases that consider dynamic graph operations. But we have not considered them in this work as none of them, to the best of our knowledge, are non-blocking nor do they satisfy \lbty.

\section*{Acknowledgement}

We are grateful to Prasad Jayanti and the anonymous referees for pointing out related work and providing helpful comments. This research was funded by the MediaLab Asia for funding Graduate Scholarship.


\bibliographystyle{ACM-Reference-Format}

\bibliography{biblio} 


\end{document}